\documentclass[twoside,leqno,twocolumn]{article}
\usepackage{ltexpprt}
 \usepackage{algorithm,algpseudocode}
\usepackage{times}
\usepackage{balance}
\usepackage{helvet}
\usepackage{amsmath}
\usepackage{caption}
\usepackage{subfig}
\usepackage{booktabs}
\usepackage{amsmath}
\usepackage{amssymb}
\usepackage{algorithmicx, algpseudocode}
\usepackage{graphicx}
\usepackage[utf8]{inputenc}
\usepackage[english]{babel}
\usepackage{xcolor}
\usepackage{slashbox}

\newtheorem{Theorem}{Theorem}

\def\CMC{\textsf{Connected Max-Cover}}

\def\TADA{\textsf{Tada-Probe}}
\def\GPM{\textsf{GPM}}
\def\TGPM{\textsf{Tada-GPM}}
\def\BGPM{\textsf{batch-GPM}}
\def\AGPM{\textsf{ada-GPM}}
\def\LINREG{\textsf{LinReg}}
\def\LOGREG{\textsf{LogReg}}

\algnewcommand\algorithmicforeach{\textbf{for each}}
\algdef{S}[FOR]{ForEach}[1]{\algorithmicforeach\ #1\ \algorithmicdo}

\begin{document}

\title{\Large Towards Optimal Strategy for Adaptive Probing in Incomplete Networks}
\author{Tri P. Nguyen, Hung T. Nguyen and Thang N. Dinh\thanks{Computer Science Department, Virginia Commonwealth University. \text{ \ \ \ \ \ \ } Email: \textit{\{trinpm,hungnt,tndinh\}@vcu.edu}} \\}
\date{}

\maketitle


\begin{abstract} \small 
	We investigate a graph probing problem in which an agent has only an incomplete view $G' \subsetneq G$ of the network and wishes to explore the network with least effort. In each step, the agent selects a node $u$ in $G'$ to probe. After probing $u$, the agent gains the information about $u$ and   its neighbors. All the neighbors of $u$ become \emph{observed} and are \emph{probable} in the subsequent steps (if they have not been probed). What is the best probing strategy to maximize the number of nodes explored in $k$ probes? This problem serves as a fundamental component for other decision-making problems in incomplete networks such as information harvesting in social networks, network crawling, network security, and viral marketing with incomplete information.

	While there are a few methods proposed for the problem,  none can perform consistently well across different network types. In this paper, we establish a strong (in)approximability for the problem, proving that no algorithm can guarantees finite approximation ratio unless P=NP. On the bright side, we design learning frameworks to capture the best probing strategies for individual network. 
	Our extensive experiments suggest that our framework can learn efficient probing strategies that \emph{consistently} outperform previous heuristics and metric-based approaches. 
\end{abstract}

\section{Introduction}

In many real-world networks, complete network topology is almost intractable to acquire, thus, most decisions are made based on incomplete networks. The impossibility to obtain complete network may stem from various sources: 1) the extreme size of the networks, e.g., Facebook, Twitter with billions of users and connections or the Internet spanning the whole planet, 2) privacy or security concerns, e.g., in Online social networks, we may not be able to see users' connections due to their own privacy settings to protect them from unwelcoming guests, 3) being hidden or undercover, e.g., terrorist networks in which only a small fraction is exposed and the rest is anonymous.

To support decision making processes based on local view and expand our observations of the networks, we investigate a network exploring problem, called \textit{Graph Probing Maximization (\GPM{})}. In \GPM{}, an agent is provided with an incomplete network $G'$ which is a subnetwork of an underlying real-world network $G \supsetneq G'$ and wants to explore $G$ swiftly through node probing. Once a node $u \in G'$ is probed, all neighbors $v \in G$ of $u$ will be observed and can be probed in the following steps. Given a budget $k$, the agent wishes to identify $k$ nodes from $G'$ to \emph{probe} to maximize the number of newly observed nodes.


Real-world applications of \GPM{} includes exploring terrorist networks to help in the planning of dismantling the network. Given an incomplete adversarial network, each suspect node can be ``probed'', e.g., getting court orders for communication record, to reveal further suspects. In cybersecurity on Online social networks (OSNs), intelligent attackers can gather users' information by sending friend requests to users \cite{Ng16}. Understanding the attacker strategies is critical in coming up with effective hardening policies. Another example is in viral marketing, from a partial observation of the network, a good probing strategy for new customers can lead to exploration of potential product sales.

While several heuristics are proposed for \GPM{} \cite{avr14y,Soundarajan15,Hanneke09,Masrour15}, they share two main drawbacks. First, they all consider selecting nodes in one batch. We argue that this strategy is ineffective as the information gained from probing nodes is not used in making the selection as shown recently \cite{Golovin10,Seeman13}. Secondly, they are metric-based methods which use a single measure to make decisions. However, real-world networks have diverse characteristics, such as different power-law degree distributions and a wide range of clustering coefficients. Thus, the proposed heuristics may be effective for particular classes of networks, but perform poorly for the others.

In this paper, we first formulate the Graph Probing Maximization and theoretically prove the strong inapproximability result that finding the optimal strategy based on local incomplete network cannot be approximated within any \textit{finite} factor. That means no polynomial time algorithm to approximate the optimal strategy within any finite multiplicative error. On the bright side, we design a novel machine learning framework that takes the local information, e.g., node centric indicators, as input and predict the best sequence of $k$ node probing to maximize the observed augmented network. We take into account a common scenario that there is available a reference network, e.g., past exposed terrorist networks when investigating an emerging one, with similar characteristics. Our framework learns a probing strategy by simulating many subnetwork scenarios from reference graph and learning the best strategy from those simulated samples.

The most difficulty in our machine learning framework is that of find the best probing strategy in sampled subnetwork scenarios from the reference network. That is how to characterize the potential gain of a \textit{probable} node in long-term (into future probing). We term this subproblem \textit{Topology-aware \GPM{}} (\TGPM{}) since both subnetwork scenario and underlying reference network are available. Therefore, we propose an $(\frac{1}{r+1})$-approximation algorithm for this problem where $r$ is the \textit{radius} of the optimal solution. Here, the radius of a solution is defined to be the largest distance from a selected node to the subnetwork. Our algorithm looks far away to the future gain of selecting a node and thus provides a nontrivial approximation guarantee. We further propose an effective heuristic improvement and study the optimal strategy by Integer Linear Programming.

Compared with metric-based methods with inconsistent performance, our learning framework can easily adapt to networks in different traits. As a result, our experiments on real-world networks in diverse disciplines confirm the superiority of our methods through consistently outperforming around 20\% the state-of-the-art methods.

Our contributions can be summarized as follows:
\begin{itemize}
	\item We first formulate the Graph Probing Maximization (\GPM{}) and show that none of existing metric-based methods consistently work well on different networks. Indeed, we rigorously prove a strong hardness result of \GPM{} problem: inapproximable within any finite factor.
    \item We propose a novel machine learning framework which looks into the future potential benefit of probing a node to make the best decision.
	\item We experimentally show that our new approach significantly and consistently improves the performance of the probing task compared to the current state-of-the-art methods on a large set of diverse real-world networks.
\end{itemize}
\textbf{Related works.}
Our work has a connection to the early network crawling literature \cite{Cho98,Chakrabarti02,Ester04,Chakrabarti99}. The website crawlers collect the public data and aim at finding the least effort strategy to gain as much information as possible. The common method to gather relevant and usable information is following the hyperlinks and expanding the domain.

Later, with the creation and explosive growth of OSNs, the attention was largely shifted to harvesting public information on these networks \cite{Chau07,Mislove07,Wwak10}. Chau et al. \cite{Chau07} were able to crawl approximately 11 million auction users on Ebay by a parallel crawler. The record of successful crawling belongs to Kwak et al. \cite{Wwak10} who gathered 41.7 million public user profiles, 1.47 billion social relations and 106 million tweets. However, these crawlers are limited to user public information due to the privacy setting feature on OSNs that protects private/updated data from unwelcoming users.

More recently, the new crawling technique of using socialbots \cite{Boshmaf12,Elishar12,Elyashar13,Fire14,Paradise15,Ng16} to friend the users and be able to see their private information. Boshmaf et al. \cite{Boshmaf12} proposed to build a large scale Socialbots system to infiltrate Facebook. The outcomes of their work are many-fold: they were able to infiltrate Facebook with success rate of 80\%; there is a possibility of privacy breaches and the security defense systems are not effective enough to prevent/stop Socialbots. The works in \cite{Elishar12,Elyashar13} focus on targeting a specific organization on Online social networks by using Socialbots to friend the organization's employees. As the results, they succeed in discovering hidden nodes/edges and achieving an acceptance rate of 50\% to 70\%.

Graph sampling and its applications have been widely studied in literature. For instance, Kim and Leskovec \cite{Kim11} study problem of inferring the unobserved parts of the network. They address network completion problem: Given a network with missing nodes and edges, how can ones complete the missing part? Maiya and Berger-Wolf \cite{Maiya10} propose sampling method that can effectively be used to infer and approximate community affiliation in the large network.

Most similar to our work, Soundarajan et al. propose MaxOutProbe \cite{Soundarajan15} probing method. MaxOutProbe estimates degrees of nodes to decide which nodes should be probed in partially observed network. This model shows better performance as compared with probing approaches based on node centralities (selecting nodes that have high degree or low local clustering). However, through experiments, we observe that MaxOutProbe's performance is still worst than the one that ranks nodes based on PageRank or Betweeness.

The probing process can be seen as a diffusion of deception in the network. Thus, it is related to the vast literature in cascading processes in the  network \cite{Nguyen10, Dinh12,Nguyen11over, Dinh14}.

\textbf{Organization}: The rest of this paper is divided into five main sections. Section~\ref{sec:model} presents the studied problem and the hardness results. We propose our approximation algorithm and machine learning model in Section~\ref{sec:algorithm}. Our comprehensive experiments are presented in Section~\ref{sec:exps} and followed by conclusion in Section~\ref{sec:con}.

\section{Problem Definitions and Hardness}
\label{sec:model}
We abstract the underlying network using an undirected graph $G = (V,E)$ where $V$ and $E$ are the sets of nodes and edges. $G$ is not completely observed, instead a subgraph $G' = (V',E')$ of $G$ is seen with $V' \subseteq V$, $E' \subseteq E$. Nodes in $G$ can be divided into three disjoint sets: $V^f, V^p$ and $V^u$ as illustrated in Fig.~\ref{fig:probing_process}.

\begin{figure}[!ht]
		\vspace{-0.2in}
		\centering
		\includegraphics[width=0.4\linewidth]{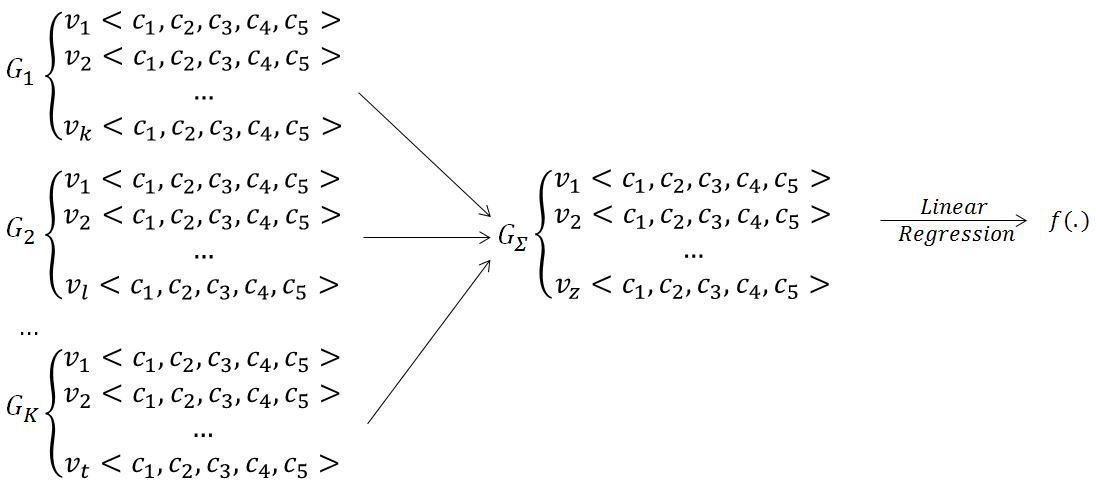}
		\vspace{-0.1in}
		\caption{Incomplete view $G'$ (shaded region) of an underlying network $G$. Nodes in $G'$ are partitioned into two disjoint subsets:  \textbf{black} nodes in $V^f$ which are fully observed and the \textbf{gray} nodes in $V^p$ which are only partially observed. Nodes outside of $G'$ are colored \textbf{white}}
		\label{fig:probing_process}
        \vspace{-0.1in}
\end{figure}

\begin{itemize}
	\item \underline{Black}/Fully observed nodes: $V^f$ contains fully observed nodes (probed) meaning all of their connections are revealed. That is if $u \in V^f$ and $(u,v) \in E$ then $(u,v) \in E'$.
    \item \underline{Gray}/Partially observed nodes: $V^p$ contains partially observed nodes $u$ that are adjacent to at least one fully probed node in $V^f$ and $u \notin V^f$. Only the connections between $u \in V^p$ and nodes in $V^f$ are observed while the others to unobserved nodes are hidden. Therefore, those nodes $u \in V^p$ become the only candidates for discovering unobserved nodes. Note that $V' = V^f \cup V^p$.
    \item \underline{White}/Unobserved nodes: $V^u = V \setminus V'$ consists of unobserved nodes. The nodes in $V^u$ have no connection to any node in $V^f$ but may be connected with nodes in $V^p$.
\end{itemize}

\emph{Node probing}: At each step, we select a candidate gray node $ u \in V^p$ to \emph{probe}. Once probed, all the neighbors of $u$ in $G$ and the corresponding edges are revealed. That is $u$ becomes a fully observed black node and each white neighbor node $v \in V^u$ of $u$ becomes gray and is also available to probe in the subsequent steps. We call the resulted graph after probing \textit{augmented graph} and use the same notation $G'$ when the context is clear. 
The main goal of \AGPM{} is to increase the size of $G'$ as much as possible.

\emph{Probing budget $k$}: In addition to the subgraph $G'$, a budget $k$ is given as the number of nodes we can probe. This budget $k$ may resemble the effort/cost that can be spent on probing. Given this budget $k$, our probing problem becomes selecting at most connected $k$ nodes that maximizes the size of augmented $G'$. Alternatively, we want to maximize the number of newly observed nodes.

We call our problem \textit{Graph Probing Maximization} (\GPM{}). There is a crucial consideration at this point: \textit{Should we select $k$ node at once or we should distribute the allowed budget in multiple steps?} The answer to this question leads to two different interesting versions: \textit{Non-Adaptive} and \textit{Adaptive}. We focus on the adaptive problem.

\begin{Definition} [Adaptive \GPM{} (\AGPM{})]
	Given an incomplete subnetwork $G'$ of $G$ and a budget $k$, the Adaptive Graph Probing Maximization problem asks for $k$ partially probed nodes in $k$ consecutive steps that maximizes the number of newly observed nodes in $G'$ at the end. 
\end{Definition}

In our \AGPM{} problem, at each step, a node is selected subject to observing the outcomes of all previous probing. This is in contrast to the non-adaptive version which asks to make $k$ selections from the initial $V^p$ at once. Thus, the adaptive probing manner is intuitively more effective than the non-adaptive counterparts. However, it is also considerably more challenging compared to non-adaptivity due to the vastly expanded search space.

\renewcommand{\arraystretch}{1.2}
\setlength\tabcolsep{5pt}
\begin{table}[hbt] \scriptsize
	\centering
	\caption{Highest performance (Perf.) metric-based methods.}
	\begin{tabular}{cc|cc|cc}
		\toprule
		\multicolumn{2}{c|}{\textbf{GnuTella}} & \multicolumn{2}{c|}{\textbf{Collaboration}} & \multicolumn{2}{c}{\textbf{Road}}  \\
		\hline
		\textbf{Top 5} & \textbf{Perf.} & \textbf{Top 5} & \textbf{Perf.} & \textbf{Top 5} & \textbf{Perf.}\\
		\midrule
		CLC & 2471 & BC & 2937 & CLC & 358 \\
		BC & 2341 & PR & 2108 & DEG & 346 \\
		CC & 1999 & CC & 2085 & BC & 346 \\
		PR & 1994 & CLC & 2061 & PR & 342 \\
		DEG & 1958 & DEG & 2048 & CC & 326 \\
		\bottomrule
	\end{tabular}
	\label{tbl:top5bnc}
    \vspace{-0.1in}
\end{table}

%

\subsection{Hardness and Inapproximability.}

\subsubsection{Empirical Observations.}
We show the inconsistency in terms of probing performance of metric-based methods, i.e., clustering coefficient (CLC), betweenness centrality (BC), closeness centrality (CC), Pagerank (PR), local degree (DEG), MaxOutProbe \cite{Soundarajan15} and random (RAND), through experiments on 3 real-world networks, i.e., Gnutella, Co-authorship and Road networks (see Sec.~\ref{sec:exps} for a detailed description). Our results are shown in Fig.~\ref{tbl:top5bnc}. From the figure, we see that the performance of metric-based methods varies significantly on different networks and none of them is consistently better than the others. For example, clustering coefficient-based method exhibits best results on Gnutella but very bad in Collaboration networks. On road networks, all methods seem to be comparable in performance.

\subsubsection{Inapproximability Result.}

Here, we provide the hardness results of the \AGPM{} problem.
Our stronger result of inapproximability
is shown in the following.

\begin{Theorem}
	\label{theo:hardness}
	\AGPM{} problem on a partially observed network cannot be approximated within any finite factor. Here, the inapproximability is with respect to the optimal solution obtained when the underlying network is known.
\end{Theorem}
To see this, we construct classes of instances of the problems that there is no approximation algorithm with finite factor.
	
	\begin{figure}[!ht]
		\centering
		\vspace{-0.05in}
		\includegraphics[width=0.34\linewidth]{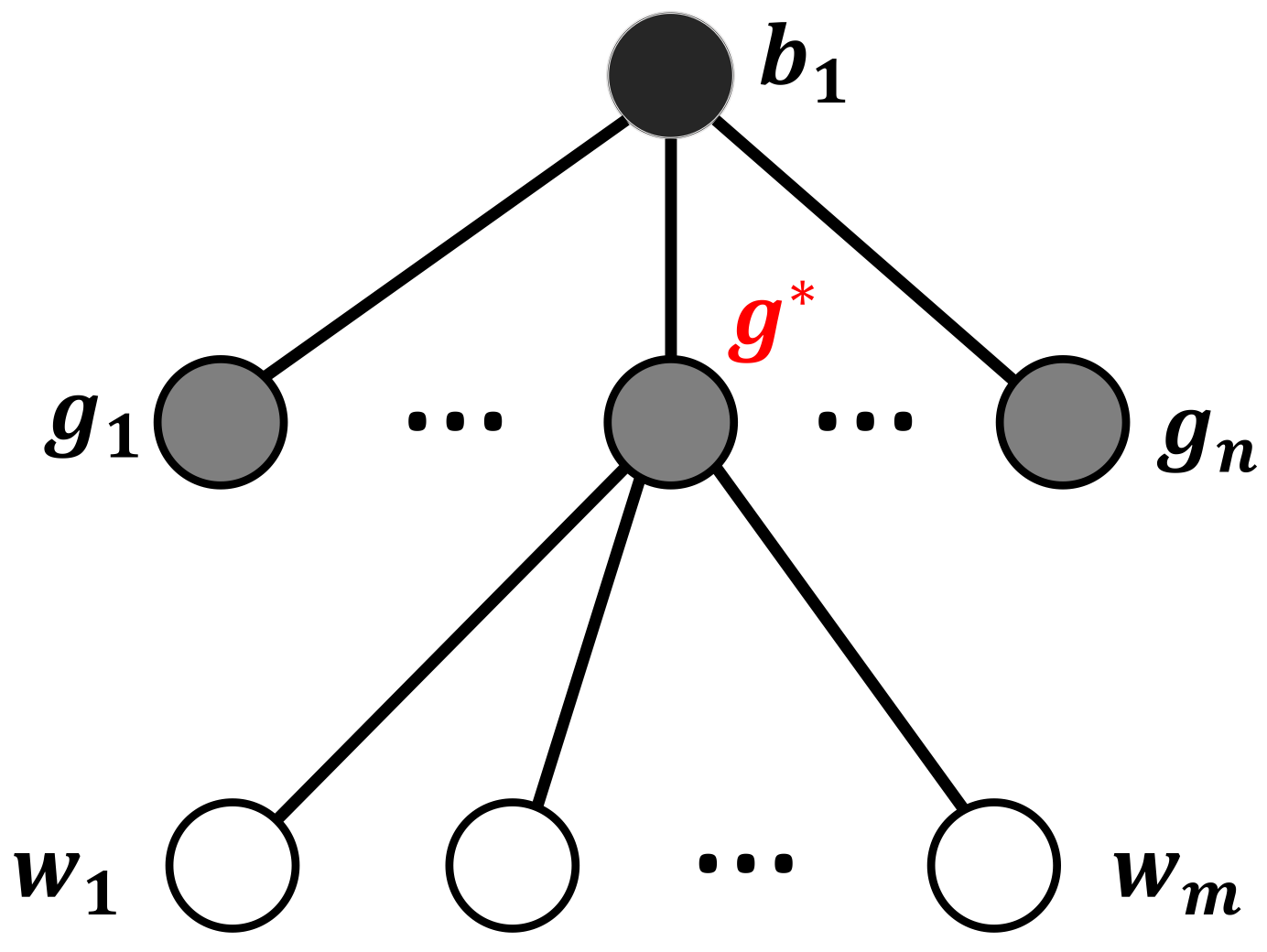}
		
        \vspace{-0.05in}
		\caption{Hardness illustration.}
		\label{fig:hard_nodegree}
        \vspace{-0.1in}
	\end{figure}
	
	We construct a class of instances of the probing problems as illustrated in Figure~\ref{fig:hard_nodegree}. Each instance in this class has: a single fully probed node in black $b_1$, $n$ observed nodes in gray each of which has an edge from $b_1$ and one of the observed nodes, namely $g^*$ varying between different instances, having $m$ connections to $m$ unknown nodes in white. Thus, the partially observed graph contains $n+1$ nodes, one fully probed and $n$ observed nodes which are selectable, while the underlying graph has in total of $n+m+1$ nodes. Each instance of the family has a different $g^*$ that $m$ unknown nodes are connected to. We now prove that in this class, no algorithm can give a finite approximate solution for the two problem.
	
	First, we observe that for any $k \geq 1$, the optimal solution which probes the nodes with connections to unknown nodes has the optimal value of $m$ newly explore nodes, denoted by $OPT = m$. Since any algorithm will not be aware of the number of connections that each gray node has, it does not know that $g^*$ leads to $m$ unobserved nodes. Thus, the chance that an algorithm selects $g^*$ is small and thus, can perform arbitrarily bad. Our complete proof is presented in the supplementary material.

\section{Learning the Best Probing Strategy}
\label{sec:algorithm}

Due to the hardness of \AGPM{} that no polynomial time algorithm with any finite approximation factor can be devised unless $P = NP$, an efficient algorithm that provides good guarantee on the solution quality in general case is unlikely to be achievable. This section proposes our machine learning based framework to tackle the \AGPM{} problem. Our approach considers the cases that in addition to the probed network, there is a reference network with similar characteristics and can be used to derive a good strategy. 

Designing such a machine learning framework is challenging due to three main reasons: \textbf{1)} what should be selected as learning samples, e.g., (incomplete) subgraphs or (gray) nodes and how to generate them? (Subsec.~\ref{subsec:build_data}); \textbf{2)} what features from incomplete subnetwork are useful for learning? (Subsec.~\ref{subsec:features}) and \textbf{3)} how to assign labels to  learning samples to indicate the benefit of selecting that node in long-term, i.e., to account for future probes?(Subsec.~\ref{subsec:features}).

\begin{figure}[!ht]
	\vspace{-0.1in}
	\centering
	\includegraphics[width=0.6\linewidth]{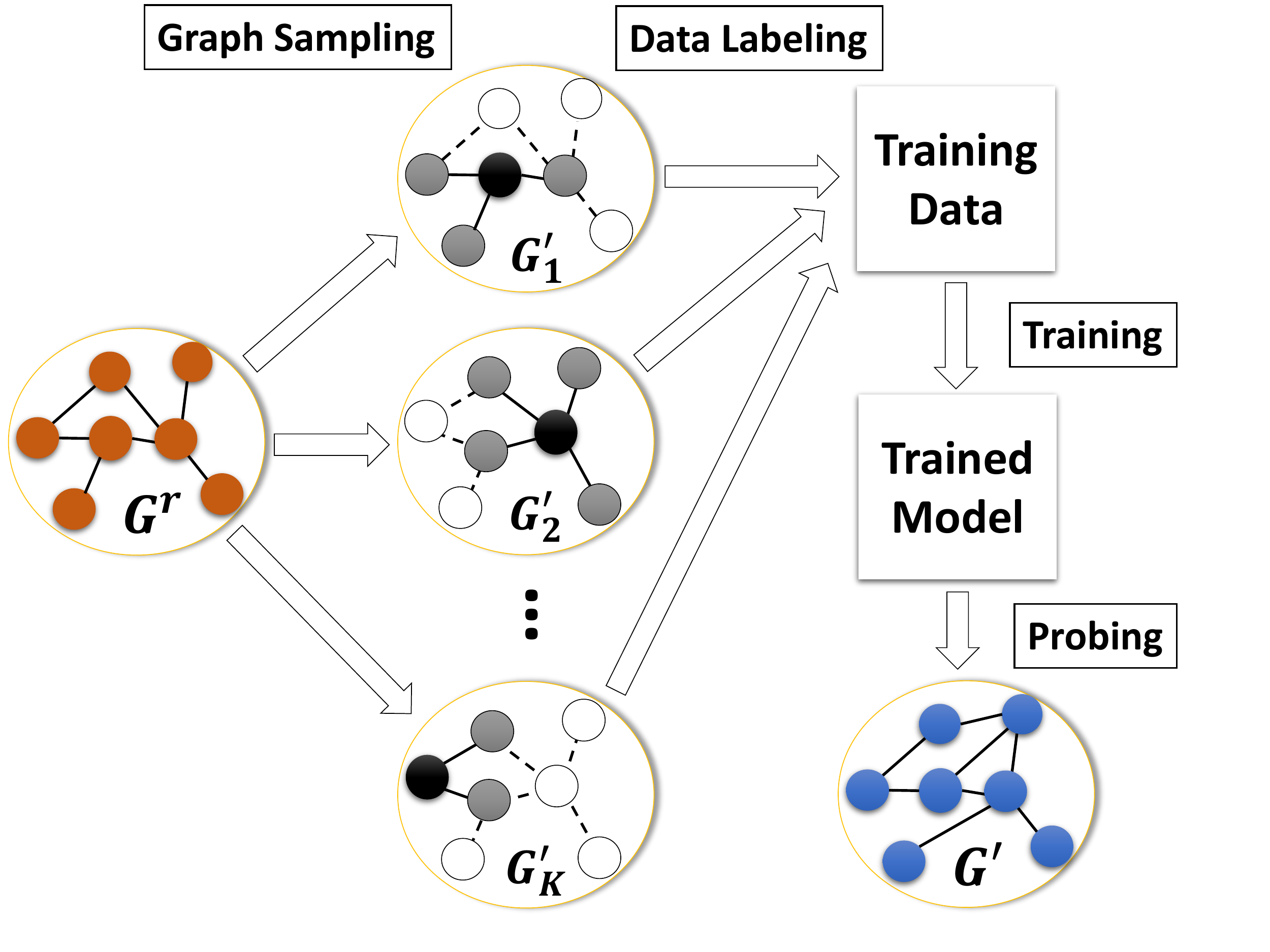}
	\vspace{-0.1in}
	\caption{Learning Framework}
	\label{fig:linear_model}
	\vspace{-0.15in}
\end{figure}
\textbf{Overview.} The general framework, which is depicted in Figure~\ref{fig:linear_model}, contains four steps: \textbf{1)} Graph sampling which generates many subnetworks from $G^r$ where each subnetwork is a sampled incomplete network with black, gray and white nodes; each candidate gray node in each sampled subnetwork creates a data point in our training data; \textbf{2)} Data labeling which labels each gray node with their long-term probing benefit; \textbf{3)} Training a model to learn the probing benefit of nodes from the features and \textbf{4)} Probing the targeted network guided by the trained machine learning model.

\subsection{Building Training Dataset.}
\label{subsec:build_data}
Given the reference network $G^r = (V^r, E^r)$, where $V^r$ is the set of $n$ nodes and $E^r$ is set of $m$ edges. Let $\mathcal{G} = \{G'_{1}, G'_{2}, ...G'_{K}\}$ be a collection of subnetwork sampled from $G^r$. The size of sampled subgraph $G'_i$ is randomly drawn between $0.5\%$ to $10\%$ of the reference graph $G^r$ following a power-law distribution. Given a subnetwork size, the sample can be generated using different mechanisms, e.g., Breadth-First-Search, Depth-First-Search or Random Walk \cite{Maiya10}.
We use $\mathcal{G}$ to construct a training data where each data point is a feature vector representing a candidate gray node. For each sample $G'_{i}, (1 \le i \le K)$, in $\mathcal{G}$, let $V_{G'_i}^p$ be the set of gray nodes in $G'_i$, we compute all the features for each node $u \in V_{G'_i}^p$ to form a data point. As such, each sample $G'_{i}$ creates $|V_{G'_i}^p|$ training data points. For assigning the label for each data point, we use our proposed algorithm \TADA{}, the heuristic improvement or the ILP algorithm (presented in Subsec.~\ref{subsec:tada}).

\subsection{Features for Learning.}
\label{subsec:features}

We select a rich set of intrinsic node features that only depend on the incomplete subnetwork and will be embedded in our learning model. Table~\ref{tbl:node_factors} shows a complete list of node features that we use in our machine learning model.

\renewcommand{\arraystretch}{1.2}
\setlength\tabcolsep{4pt}
\begin{table}[hbt]\scriptsize
	\vspace{-0.1in}
	\centering
	\caption{Set of features for learning}
   \vspace{-0.1in}
	\begin{tabular}{p{1cm}p{6.5cm}}
		\addlinespace
		\toprule
		\textbf{Factor}  & \textbf{Description}  \\
		\midrule
		$BC$ & Betweenness centrality score \cite{Newman10} of $u$ in $G'$ \\
		$CC$ & Closeness centrality score \cite{Newman10} of $u$ in $G'$ \\
		$EIG$ & Eigenvalue centrality score \cite{Newman10} of $u$ in $G'$\\	
		$PR$ & Pagerank centrality score \cite{Newman10} of $u$ in $G'$ \\
		$Katz$ & Katz centrality score \cite{Newman10} of $u$ in $G'$ \\
		$CLC$ & Clustering coefficient score of $u$ in $G'$ \\
		$DEG$ & Degree of $u$ in $G'$ \\		
		$BNum$ & Number of black nodes in $G'$ \\
		$GNum$ & Number of gray nodes in $G'$ \\
		$BDeg$ & Total degree of black nodes in $G'$ \\
		$BEdg$ & Number of edges between black nodes in $G'$ \\
		\bottomrule
	\end{tabular}
	\label{tbl:node_factors}
    \vspace{-0.2in}
\end{table}

\subsection{An $\frac{1}{r+1}$-Approximation Algorithm for \TGPM{}.}
\label{subsec:tada}

We first propose an $\frac{1}{r+1}$-optimal strategy for \TGPM{} to probe a sampled subnetwork of the reference network, where $r$, called \textit{radius}, is the largest distance from a node in the optimal solution to the initially observed network. This algorithm assigns labels for the training data.

An intuitive strategy, called Naive Greedy, is to select node with highest number of connections to unseen nodes. Unfortunately, this strategy can be shown to perform arbitrarily bad by a simple example. The example includes a fully probed node having a connection to a degree-2 node, which is a bridge to a huge component which is not reachable from the other nodes, and many other connections to higher degree nodes. Thus, the Naive Greedy will not select the degree-2 node and never reach the huge component.



Our algorithm is inspired by a recent theoretical result for solving the \textit{Connected Maximum Coverage} (\CMC{}) problem in \cite{Vandin11}. The \CMC{} assumes a set of elements $\mathcal{I}$, a graph $G = (V,E)$ and a budget $k$. Each node $v \in V$ associates with a subset $P_v \subseteq \mathcal{I}$. The problem asks for a \textit{connected} subgraph $g$ with at most $k$ nodes of $G$ that maximizes $\cup_{v \in g} P_v$. The proposed algorithm in \cite{Vandin11} sequentially selects nodes with highest ratio of newly covered nodes to the length of the shortest path from observed nodes and is proved to obtain an $\frac{e-1}{(2e-1)r}$-approximation factor.
\vspace{-0.1in}
\setlength{\textfloatsep}{3pt}
\begin{algorithm} \small
	\caption{\small \TADA{} Approximation Algorithm}
	\label{alg:agpm}
	\textbf{Input}: The reference network $G = (V,E)$, a sampled subnetwork $G' = (V',E')$ and a budget $k$. \\
	\textbf{Output}: Augmented graph of $G'$.
	\begin{algorithmic}[1]
		\State Collapse all fully observed nodes to a single root node $R$
		\State $i = 0$
		\While{$i < k$}
		\State $v_{max} \leftarrow \max_{v \in V \backslash V^f, |P_{V^f}(v)| \leq k-i} \frac{|O(v) \backslash V'|}{|P_{V^f}(v)|}$
		\State Probe all the nodes $v \in P_{V^f}(v_{max})$
		\State Update $V', V^f, V^p, V^u$ and $E'$ accordingly
		\State $i = i + |P_{V^f}(v_{max})|$
		\EndWhile
		\State \textbf{return} $G' = (V', E')$
	\end{algorithmic}
\end{algorithm}
\vspace{-0.1in}

Each node in a network can be viewed as associated with a set of connected nodes. We desire to select $k$ connected nodes to maximize the union of $k$ associated sets. However, different from \CMC{} in which any connected subgraph is a feasible solution, \TGPM{} requires the $k$ selected nodes to be connected from the fully observed nodes $V^f$. Thus, we analogously put another constraint of connectivity to a fixed set of nodes on the returned solution and that adds a layer of difficulty. Interestingly, we show that rooting from observed nodes and greedily adding nodes in the same manner as \cite{Vandin11} gives an $\frac{1}{r+1}$-approximate solution. Additionally, our analysis leads to a better approximation result for \CMC{} since $\frac{e-1}{(2e-1)r} < \frac{1}{2.58\cdot r} < \frac{1}{r+1}$.

Let denote $O(v)$ to be the set of nodes that $v$ is \textit{connected to}, i.e., $ O(v)= \{u | (v,u) \in E \}$ and $P_{V^f}(v)$ be the set of nodes on the shortest path from nodes in $V^f$ to $v$. For a set of nodes $S$, we call $f(S)$ the number of newly discovered nodes by probing $S$. Hence, $f(S)$ is our objective function. Our approximation algorithm, namely \textit{Topology-aware Adaptive Probing} (\TADA{}), is described in Alg.~\ref{alg:agpm}.

The algorithm starts by collapsing all fully observed nodes in $G'$ into a single root node $R$ which serves as the starting point. It iterates till all $k$ alloted budget have been exhausted into selecting nodes (Line~3). In each iteration, it selects a node $v_{max} \in V \backslash V^f$ within distance $k-i$ having maximum ratio of the number of unobserved nodes $|O(v)\backslash V'|$ to the length of shortest path from nodes in $V^f$ to $v$ (Line~4). Afterwards, all the nodes on the shortest path from $V^f$ to $v_{max}$ are probed (Line~5) and the incomplete graph is updated accordingly (Line~6).

%
%

The approximation guarantee of \TADA{} is stated in the following theorem.

\begin{Theorem}
	\TADA{} returns an $\frac{1}{r+1}$-approximate probing strategy for \TGPM{} problem where $r$ is the radius of the optimal solution.
\end{Theorem}
\begin{proof}
	Let denote the solution returned by \TADA{} $\hat S$ and an optimal solution $S^* = \{v^*_1,v^*_2,\dots,v^*_k\}$ which results in the maximum number of newly discovered nodes, denoted by $OPT$. We assume that both $\hat S$ and $S^*$ contain exactly $k$ nodes since adding more nodes never gives worse solutions. We call the number of additional unobserved nodes discovered by $S'$ in addition to that of $S$, denoted by $\Delta_{S}(S')$, the \textit{marginal benefit} of $S'$ with respect to $S$. For a single node $v$, $\Delta_{S}(v) = \Delta_{S}(\{v\})$. In addition, the ratio of the marginal benefit to the distance from the set $S$ to a node $v$, called \textit{benefit ratio}, is denoted by $\delta_{S}(v) = \frac{\Delta_{S}(v)}{|P_{S}(v)|}$.
	
	Since in each iteration of \TADA{}, we add all the nodes along the shortest path connecting $V^f$ to $v_{max}$, we assume that $t \leq k$ iterations are performed. In iteration $i \geq 0$, node $v^i_{max}$ is selected to probe and, up to that iteration, $S^i$ nodes have been selected so far.
	
	Due to the greedy selection, we have, $\forall i \geq 1, \forall \hat v \in S^i\backslash S^{i-1}, \forall v^* \in S^*$,
	\vspace{-0.1in}
	\begin{align}
	\delta_{S^{i-1}}(v^i_{max}) \geq \delta_{S^{i-1}}(v^*)
	\end{align}
	\vspace{-0.2in}
	
	\noindent Thus, we obtain,
	\vspace{-0.2in}
	\begin{align}
	|P_{S^{i-1}}(v^i_{max})|\cdot \delta_{S_{i-1}}(v^i_{max}) \geq \sum_{j = |S^{i-1}|}^{|S^{i}|} \delta_{S^{i-1}}(v^*_j)
	\end{align}
	\vspace{-0.25in}
	
	\noindent or, equivalently,
	\vspace{-0.2in}
	\begin{align}
	\label{eq:iter}
	\Delta_{S^{i-1}(v^i_{max})} \geq \sum_{j = |S^{i-1}|}^{|S^{i}|} \delta_{S^{i-1}}(v^*_j)
	\end{align}
	\vspace{-0.15in}
	
	Adding Eq.~\ref{eq:iter} over all iterations gives,
	\vspace{-0.1in}
	\begin{align}
	\label{eq:iter_2}
	\sum_{i = 1}^{t}\Delta_{S^{i-1}(v^i_{max})} \geq \sum_{i = 1}^{t}\sum_{j = |S^{i-1}|}^{|S^{i}|} \delta_{S^{i-1}}(v^*_j)
	\end{align}
	\vspace{-0.15in}
	
	\noindent The left hand side is actually $f(\hat S)$ which is sum of marginal benefit over all iterations. The other is the sum of benefit ratios over all the nodes in the optimal solution $S^*$ with respect to sets $S^{i-1}$, where $0 \leq i \leq k-1$, which are subsets of $\hat S$. Thus, $\forall i,j$,
	\vspace{-0.1in}
	\begin{align}
	\delta_{S^{i-1}}(v^*_j) \geq \delta_{\hat S}(v^*_j) \geq \frac{\Delta_{\hat S}(v^*_j)}{|P_{\hat S}(v^*_j)|} \geq \frac{\Delta_{\hat S}(v^*_j)}{r}
	\end{align}
	\vspace{-0.15in}
	
	\noindent Then, the right hand side is,
	\vspace{-0.1in}
	\begin{align}
	\label{eq:theo1_bound}
	\sum_{i = 1}^{t}\sum_{j = |S^{i-1}|}^{|S^{i}|} \delta_{S^{i-1}}(v^*_j) \geq \sum_{i = 1}^{t}\sum_{j = |S^{i-1}|}^{|S^{i}|} \frac{\Delta_{\hat S}(v^*_j)}{r}
	\end{align}
	\vspace{-0.1in}
	
	\noindent Notice that $\Delta_{\hat S}(v^*_j)$ is the marginal benefit of node $v^*_j$ with respect to set $\hat S$, hence, the summation itself becomes,
	\vspace{-0.15in}
	\begin{align}
	\sum_{i = 1}^{t}\sum_{j = |S^{i-1}|}^{|S^{i}|} \Delta_{\hat S}(v^*_j) = \sum_{j = 1}^{k} \Delta_{\hat S}(v^*_j) & = f(S^*) - f(\hat S) \nonumber \\
	& = OPT - f(\hat S) \nonumber
	\end{align}
	\vspace{-0.3in}
	
	\noindent Thus, Eq.~\ref{eq:iter_2} is reduced to,
	\vspace{-0.1in}
	\begin{align}
	f(\hat S) \geq \frac{OPT - f(\hat S)}{r}
	\end{align}
	\vspace{-0.2in}
	
	\noindent Rearranging the above equation, we get,
	\vspace{-0.05in}
	\begin{align}
	f(\hat S) \geq \frac{OPT}{r+1}
	\end{align}
	\vspace{-0.2in}
	
	\noindent which completes our proof.
\end{proof}

    \vspace{-0.05in}
\subsubsection{Improved Heuristic.}
Despite the $\frac{1}{r+1}$-approximation guarantee, \TADA{} algorithm only considers the gain of ending node in a shortest path and completely ignores the on-the-way benefit. That is the newly observed nodes discovered when probing the intermediate nodes on the shortest paths are neglected in making decisions. Thus, we can improve \TADA{} by counting all the newly observed nodes \textit{along the connecting paths} which are not necessarily the shortest paths and the selection criteria of taking the path with largest benefit ratio is applied. Since the selected path of nodes has the benefit ratio of at least as high as that of considering only the ending nodes, the $\frac{1}{r+1}$-approximation factor is preserved.

Following that idea, we propose a Dijkstra based algorithm to select the path with maximum benefit ratio. We assign for each node $u$ a benefit ratio $\delta(u)$, a distance measure $d(u)$ and a benefit value $\Delta(u)$. Our algorithm iteratively selects a node $u$ with highest benefit ratio and propagates the distance and benefit to its neighbors: if neighbor node $v$ observes that by going through $v$, $u$'s benefit ratio gets higher, $v$ will update its variables to have $u$ as the direct predecessor. Our algorithm finds the path with highest benefit ratio.

Note that extra care needs to be taken in our algorithm to avoid having \textit{loops} in the computed paths. The loops normally do not appear since closing a loop only increases the distance by one while having the same benefit of the path. However, in extreme cases where a path passes through a node with exceptionally high number of connections to unobserved nodes, loops may happen. To void having loops, we check whether updating a new path will close a loop by storing the predecessor of each node and traverse back until reaching a fully observed node.

    \vspace{-0.09in}
\subsubsection{Optimal ILP Algorithm.}
To study the optimal solution for our \TGPM{} problem when topology is available, we present our Integer Linear Programming formulation. Hence, We can use a convenient off-the-shelf solver, e.g., Cplex, Gurobi, to find an optimal solution. Unfortunately, Integer Linear Programming is not polynomially solvable and thus, extremely time-consuming.

In the prior step, we also collapse all fully probed nodes into a single node $r$ and have connections from $r$ to partially observed nodes. Assume that there are $n$ nodes including $r$, for each $u \in V$, we define $y_u \in \{0,1\}, \forall u \in V$ such that,
\vspace{-0.05in}
\begin{align}
y_u = \left\{ \begin{array}{ll}
1 & \text{ if node } u \text{ is observed,}\\
0 & \text{ otherwise}.
\end{array}
\right. \nonumber
\end{align}
\vspace{-0.15in}

\noindent Since at most $k$ nodes are selected adaptively, we can view the solution as a tree with at most $k$ layers. Thus, we define $x_{uj} \in \{0,1\}, \forall u \in V, j = 1..k$ such that,
\vspace{-0.05in}
\begin{align}
x_{uj} = \left \{ \begin{array}{ll}
1 & \text{ if node } u \text{ is selected at layer }j \text{ or earlier},\\
0 & \text{ otherwise}.
\end{array}
\right. \nonumber
\end{align}
\vspace{-0.15in}

\noindent The \TGPM{} problem selects at most $k$ nodes, i.e., $\sum_{u \in V} x_{uk} \leq k$ to maximize the number of newly observed nodes, i.e., $\sum_{u \in V} y_{u}$. A node is observed if at least one of its neighbors is selected meaning $y_u \leq \sum_{v \in N(u)} x_{vk}$ where $N(u)$ denotes the set of $u$'s neighbors. Since $r$ is the initially fully observed nodes, we have $x_{r0} = 1$. Furthermore, $u$ can be selected at layer $j$ if at least one of its neighbors has been probed and thus, $x_{uj} \leq \sum_{v \in N(u)} x_{v(j-1)}$.

Our formulation is summarized as follows,
\vspace{-0.05in}
\begin{align}
\label{eq:agpm_ip}
\max \quad & \sum_{u \in V} y_{u} - |N(r)| - 1
\end{align}
\vspace{-0.3in}
\begin{align}
\text{ s.t.}\quad \quad & x_{r0} = 1, x_{u0} = 0, \qquad u \in V, u \neq r \nonumber \\
& \sum_{u \in V} x_{uk} \leq k \nonumber \\
& y_u \leq \sum_{v \in N(u)} x_{vk}, \qquad \forall u \in V, \nonumber\\
& x_{uj} \leq x_{u(j+1)}, \qquad \forall u \in V, j = 0..k-1, \nonumber\\
& x_{uj} \leq x_{u(j-1)} + \sum_{v \in N(u)} x_{v(j-1)}, \text{ } \forall u \in V, j = 1..k, \nonumber\\
& x_{uj}, y_{u} \in \{0, 1\}, \forall u \in V, j = 0..k. \nonumber
\end{align}
\vspace{-0.2in}

From the solution of the above ILP program, we obtain the solution for our \TGPM{} instance by simply selecting nodes $u$ that $x_{uk} = 1$. Note that the layering scheme in our formulation guarantees both the connectivity of the returned solution and containing root node $r$ and thus, the returned solution is feasible and optimal.
\begin{Theorem}
	The solution of the ILP program in Eq.~\ref{eq:agpm_ip} infers the optimal solution of our \TGPM{}.
\end{Theorem}

\begin{figure}[!ht]
	\centering
	\includegraphics[width=0.65\linewidth]{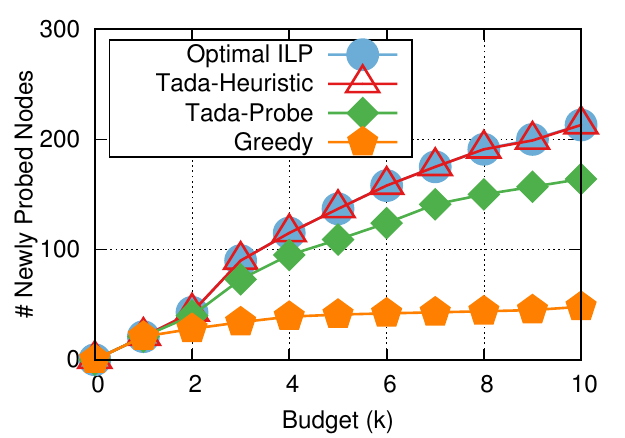}
	\vspace{-0.1in}
	\caption{Performance of different algorithms}
	\label{fig:compare_alg}
\end{figure}

    \vspace{-0.09in}
\subsubsection{Empirical Evaluation.}
Here, we compare the probing performance in terms of the number of newly probed nodes that different algorithms deliver on a Facebook ego network\footnote{http://snap.stanford.edu/data/egonets-Facebook.html} with 347 nodes and 5029 edges. The results are presented in Fig.~\ref{fig:compare_alg}. The figure shows that our Heuristic improvement very often meets the optimal performance of the ILP algorithm while Tada-Probe is just below the former two methods. However, the naive Greedy algorithm performs badly due to having no guarantee on solution quality.

\subsection{Training Models.}

We consider two classes of well-studied machine learning models. First, linear regression model is applied to learn a linear combination of features characterized by coefficients. These coefficients learn a linear representation of the dependence of labels on the features. The output of training phase of linear regression model is a function $f_{Lin}(.)$ which will be used to estimate the gain $w_{u}^o$ of probing a node $u$ in subgraph $G'$. It is noted that we learn $f_{Lin}(.)$ from sampled graph $G'_i$ of the reference graph $G^r$ and use $f_{Lin}(.)$ to probe the incomplete $G'$. 

Secondly, we consider logistic regression to our problem as follow: Let $w_u^o$, $w_v^o$  be the gain of probing the two nodes $u$ and $v$ respectively. Given a pair of nodes $<u,v>$ in $V_{G'_i}^p$, our logistic model $f_{Log(.)}$ learns to predict whether $w_u^o$ is larger than $w_v^o$. Thus, for each $G'_{i}$ in $\mathcal{G}$, we compute node features of each node $u \in V_{G'_i}^p$. We generate $\binom{|V_{G_{i}}^p|}{2}$ pairs of nodes for each subgraph $G'_{i}$ and then concatenate features of two nodes in a pair to form a single data point. Each data point $<u,v>$ is labeled by binary value ($1$ or $0$):
\vspace{-0.05in}
\begin{align}
	l = \left \{ \begin{array}{ll} 
	1 & \mbox{if $w_u^o \geq w_v^o$};\\
	0 & \mbox{if $w_u^o < w_v^o$}.\end{array} \right.
\end{align}

\section{Experiments}
\label{sec:exps}

In this section, we perform  experiments on real-world networks to evaluate performance of the proposed  methods.


\renewcommand{\arraystretch}{1.2}
\setlength\tabcolsep{2pt}
\begin{table}[hbt] \small
	\centering
	\caption{Statistics for the networks used in our experiments. ACC stands for Average Clustering Coefficient. Bold and underlined networks are used for training. }
	\vspace{-0.2in}
	\begin{tabular}{llllc}
		\addlinespace
		\toprule
		\textbf{Name}  & \textbf{Network Type} & \#\textbf{Node} & \#\textbf{Edges} & \textbf{ACC} \\ 
		\midrule 
		Roadnw-CA & Road & $21k$ & $21k$ & $7.10^{-5}$\\
		Roadnw-OL & Road & $6k$ & $7k$ & $0.01$\\
		\textbf{\underline{Roadnw-TG}}  & Road & \textbf{$18k$} & \textbf{$23k$} & \textbf{$0.018$}\\
		GnuTella04  & p2p (GnuTella) & $11k$ & $40k$ & $0.006$ \\
		GnuTella05 & p2p & $9k$ & $32k$ & $0.007$ \\
		\textbf{\underline{GnuTella09}} & p2p & \textbf{$8k$} & $26k$ & $0.009$\\
		Ca-GrQc & Collaboration (CA) & $5k$ & $14k$ & $0.529$\\
		Ca-HepPh & Collaboration & $12k$ & $118k$ & $0.611$\\
		Ca-HepTh & Collaboration & $10k$ & $26k$ & $0.471$\\
		Ca-CondMat & Collaboration & $23k$ & $93k$ & $0.633$\\
		\textbf{\underline{Ca-AstroPh}} & Collaboration & $18k$ & $198k$ & $0.630$\\
		\bottomrule
	\end{tabular}
	\label{tbl:dataset}
	\vspace{-0.1in}
\end{table}

\begin{figure*}[!ht]
	\centering
	\subfloat[GnuTella Network]{
		\includegraphics[width=0.3\linewidth]{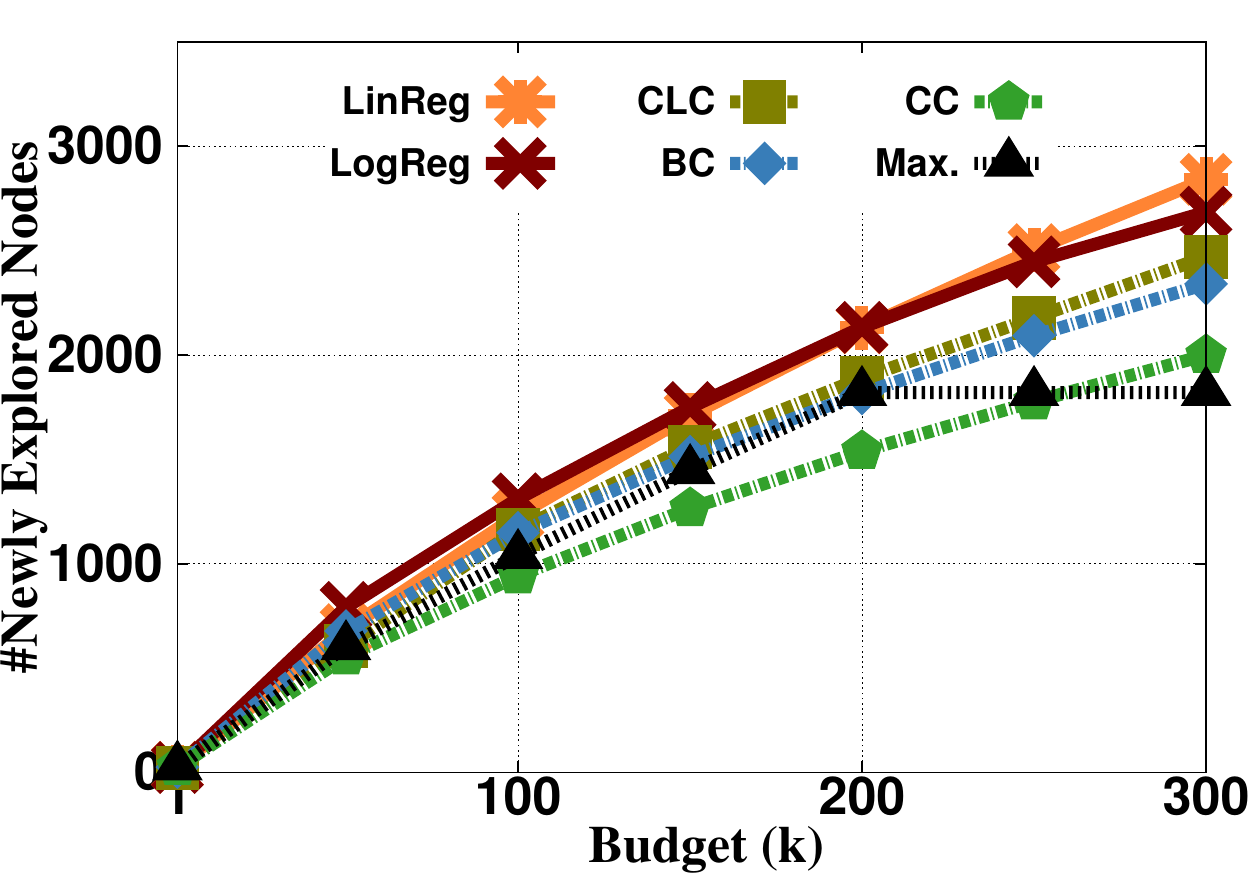}
	}
	\subfloat[Collaboration Network]{
		\includegraphics[width=0.3\linewidth]{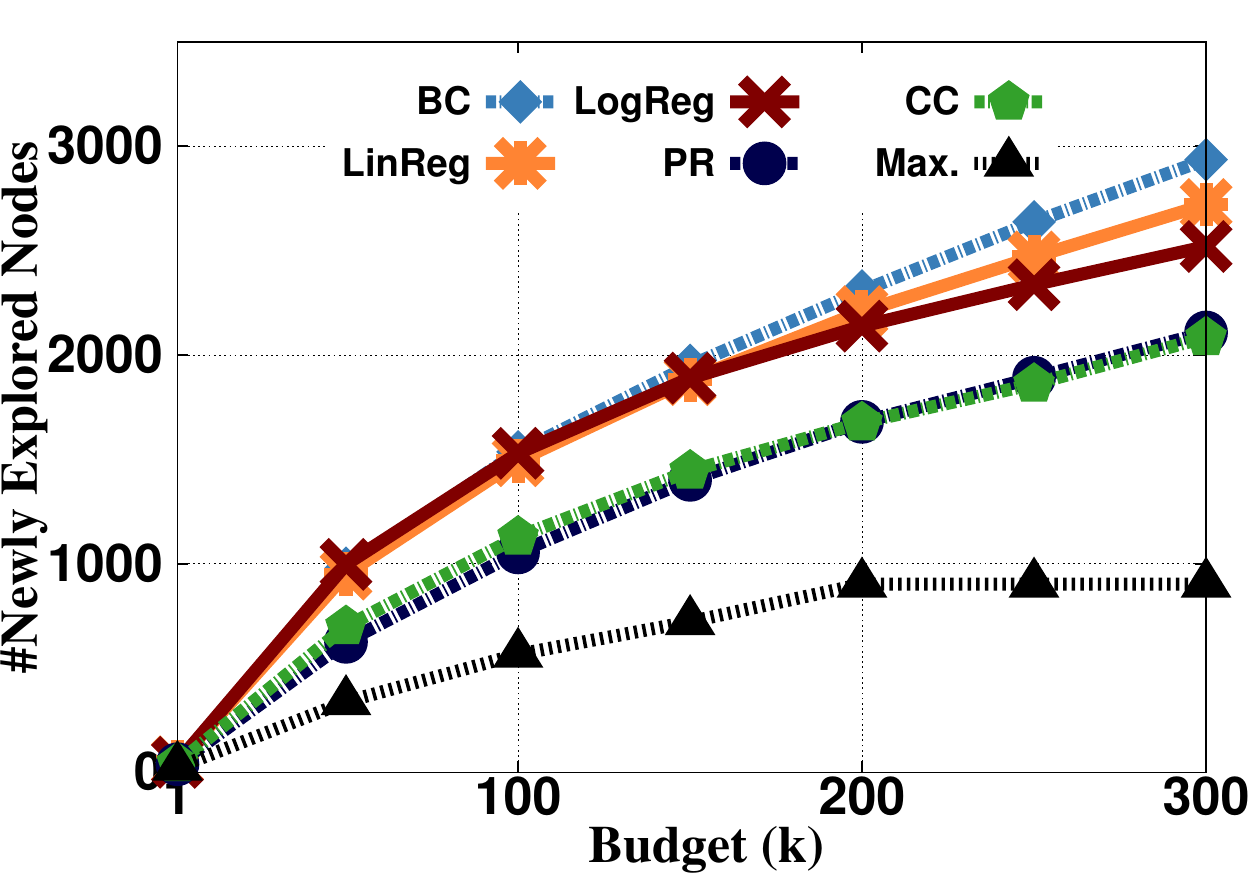}
	}
	\subfloat[Road Network]{
		\includegraphics[width=0.3\linewidth]{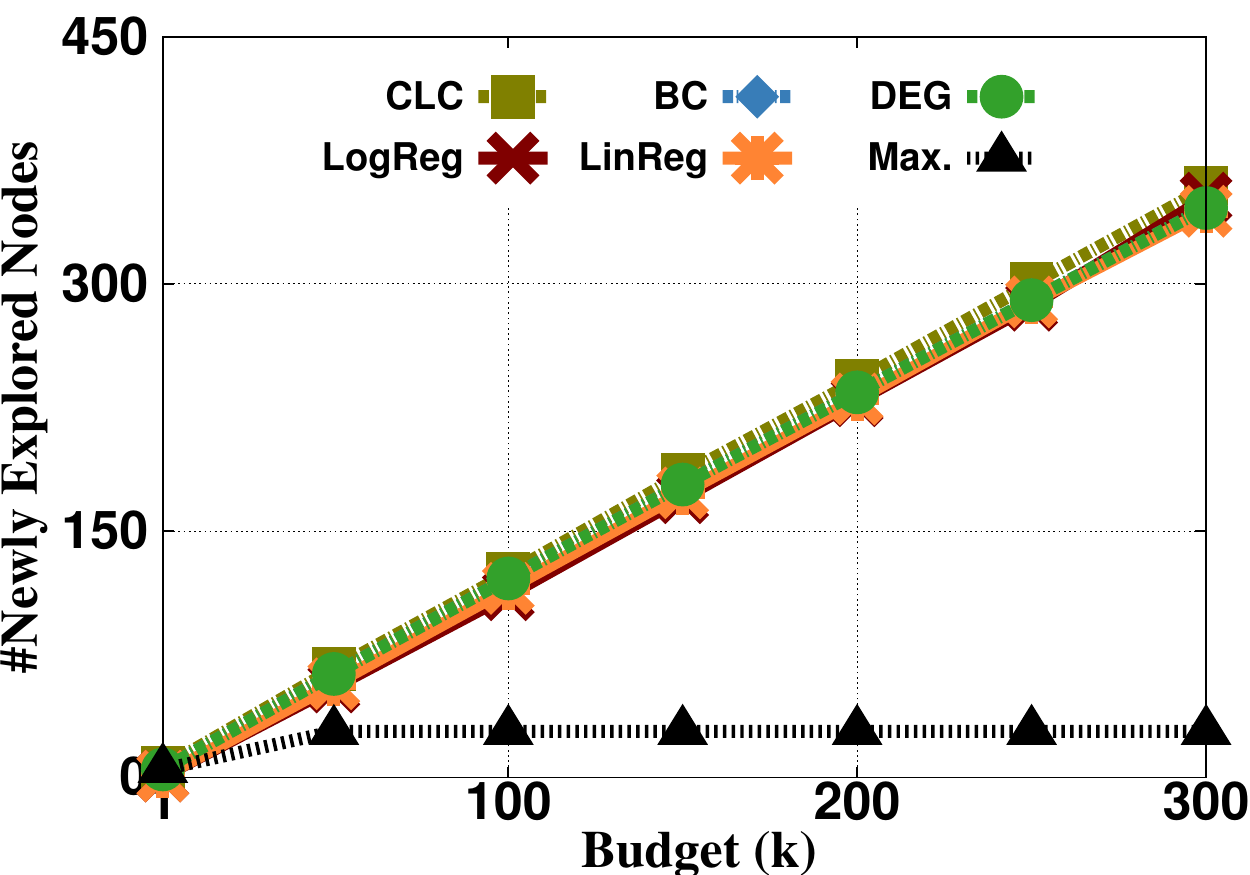}
	}
	\vspace{-0.1in}
	\caption{Performance Comparison of Base Case Machine Learning ($h=1$) with Metric-based and Heuristic Approaches.}
	\label{fig:set2}
	\vspace{-0.25in}
\end{figure*}

\begin{figure*}[!ht]
	\centering
	
	\subfloat[GnuTella Network]{
		\includegraphics[width=0.3\linewidth]{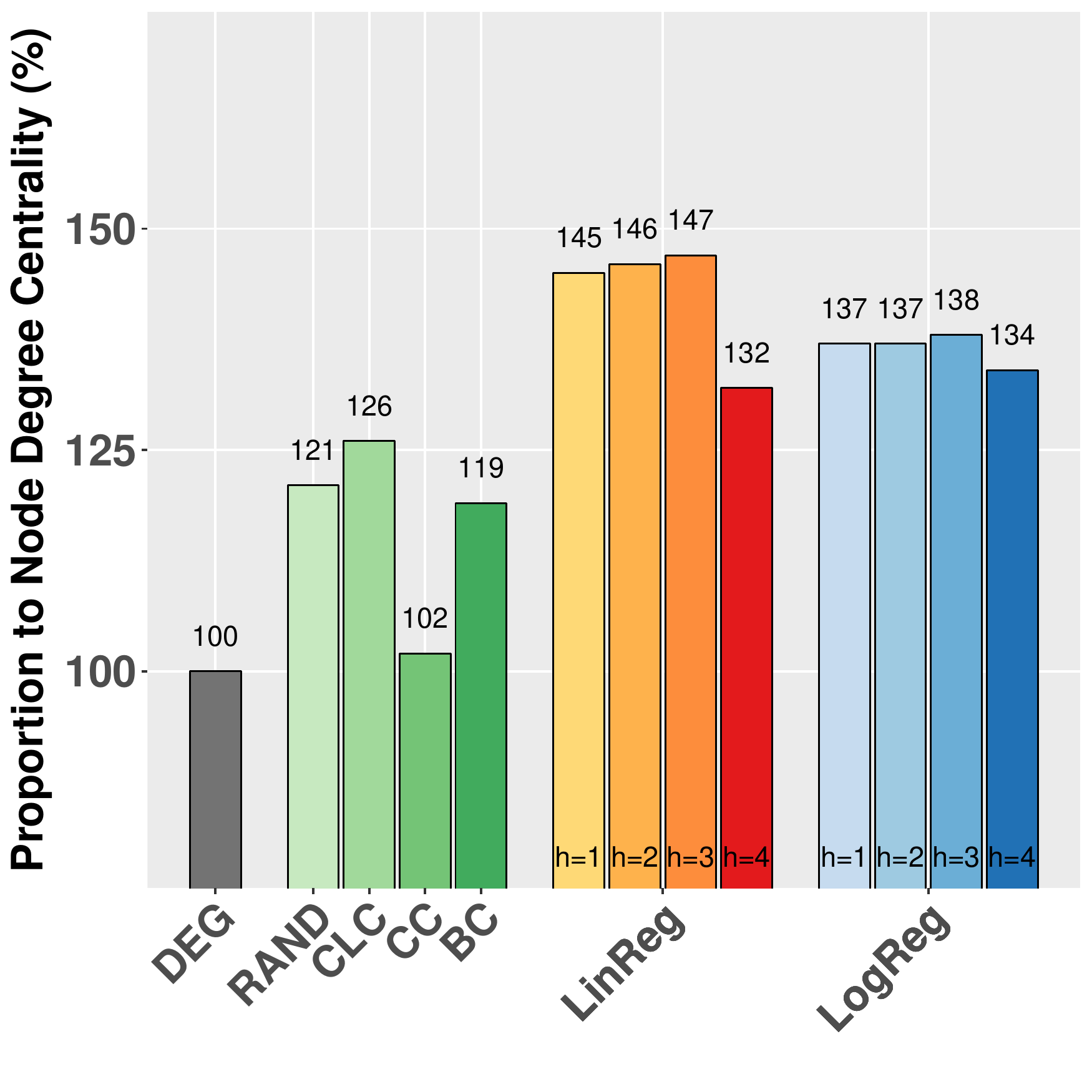}
	}
	\subfloat[Collaboration Network]{
		\includegraphics[width=0.3\linewidth]{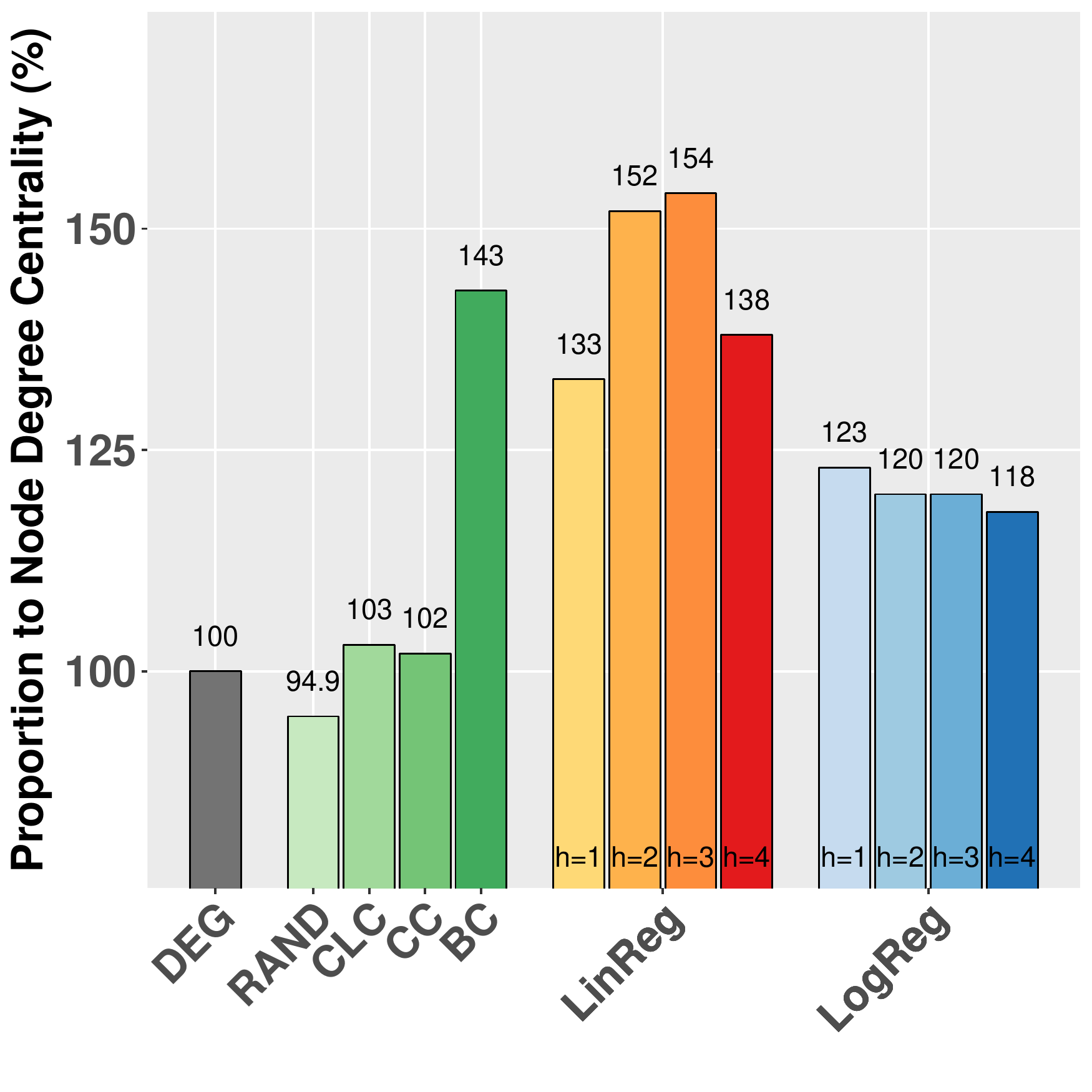}
	}
	\subfloat[Road Network]{
		\includegraphics[width=0.3\linewidth]{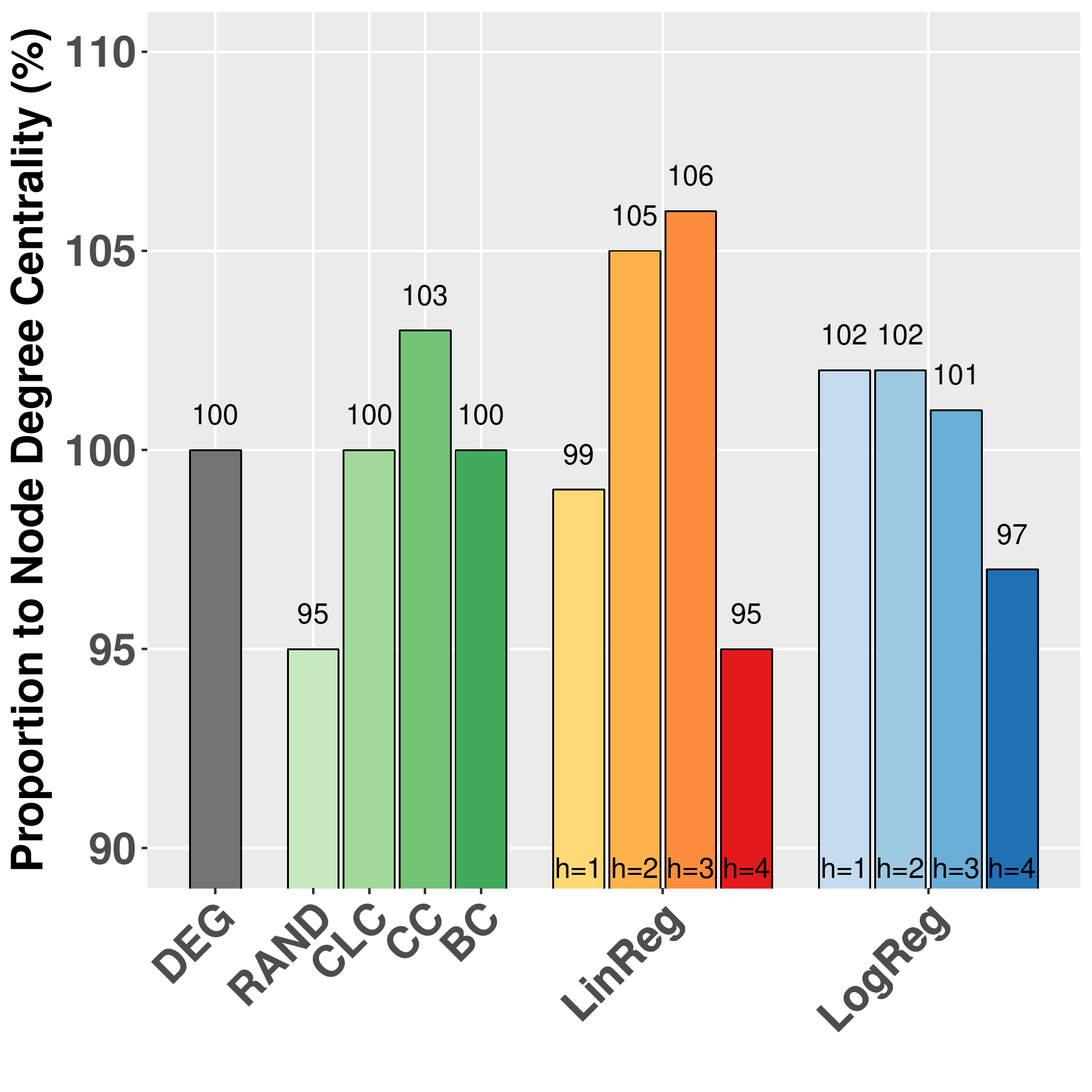}
	}
	\vspace{-0.1in}
	\caption{Performance of Machine Learning Method with Further Looking Ahead and Other Probing Methods.}
	\label{fig:set3} 
	\vspace{-0.25in}
\end{figure*}


\textbf{Datasets.}
Table~\ref{tbl:dataset} describes three types of real-world networks used in our experiments. The road network \cite{li05, brinkhoff02} includes edges connecting different points of interest (gas stations, restaurants) in cites. 
The second type of network includes several snapshots of GnuTella network with nodes represent hosts in GnuTella network topology and edges represent connections between hosts. 
In the third type, we use five collaboration networks that cover scientific collaborations between authors papers in which nodes represent scientists, edges represent collaborations. 
The last type of network models metabolic pathways which are linked series of chemical reactions occurring in cell.

\renewcommand{\arraystretch}{1.2}
\setlength\tabcolsep{2pt}
\begin{table}[hbt] \small
    \vspace{-0.09in}
	\centering
	\caption{Number of new explored nodes at budget $k = 300$ from all implemented probing methods in all datasets.}    
	\begin{tabular}{lllc}
		\addlinespace
		\toprule	
		\backslashbox[1mm]{\textbf{Methods}}{\textbf{Dataset}} & \textbf{GnuTella Net.} & \textbf{CA Net.} & \textbf{Road Net.} \\		
		\midrule 
		CLC & \textbf{\textcolor{red}{2471}} & 3098 & \textbf{\textcolor{red}{358}} \\
		BC & 2341 & \textbf{\textcolor{red}{5052}} & 346 \\
		DEG & 1958 & 3471 & 346\\
		CC & 1999 & 3547 & 326 \\
		PR & 1994 & 3639 & 342 \\
		RAND & 2381 & 2911 & 329 \\
		MaxOutProbe & 1820 & 902 & 28\\	
		\bottomrule
	\end{tabular}
	\label{tbl:bnc}
\end{table}

\textbf{Sampling Methods.}
We adopt the Breadth-First Search sampling ~\cite{Maiya10} to generate subgraphs  for our experiments. In our regression model, for each network that is marked for training in Table ~\ref{tbl:dataset}, we generate samples with number of nodes varies from $0.5\%$ to $10\%$ number of nodes in $G$. The size distribution of the subgraphs follows a power law distribution with power-exponent $\gamma=-1/4$. 
The size of subgraphs used in validation is kept to be roughly  $5\%$ of the network size.

\textbf{Probing Methods.}
We compare performance of our linear regression (\LINREG{}) and logistic regression (\LOGREG{}) probing method with MaxOutProbe\cite{Soundarajan15},  RAND (probe a random gray node),  and centrality-based methods DEG, BC, CC, PR, CLC (see table~\ref{tbl:node_factors}) that probe the node with the highest centrality value in each step. Each considered method has different objective function for selecting nodes to probe. Sharing the same idea, all of these methods rank nodes in set $V^p$ of $G'$, and select the nodes with highest ranking to probe first.



\textbf{Performance Metrics.} For each probing method, we conduct probing at budgets $k \in \{1, 100, 200, 300\}$. During probing process, we compare performance among probing methods using \emph{newly explored nodes}, i.e., the increase in the number of nodes in $G'$ after $k$ probes. For each of the network used in the validation, we report the average results over 50 subgraphs. For obvious reason, we do not use the training network for validation. 

We use statistical programming language R to implement MaxOutProbe and use C++ with igraph framework to implement \LINREG{}, \LOGREG{} and metric-based probing methods. All of our experiments are run on a Linux server with a 2.30GHz Intel(R) Xeon(R) CPU and 64GB memory.

\subsection{Comparison between Machine Learning Methods and Metric-based Probing Methods.}


Table \ref{tbl:bnc} presents probing performance of metric-based methods in each group of networks with the best one is highlighted. Reported result is the number of explored nodes at the end of probing process ($k = 300$).

In Fig. \ref{fig:set2}, we take top 3 metric-based methods in each type of network and compare their performance with \LINREG{} and \LOGREG{}. Here, both machine learning methods are trained by using the benefit of node in only 1-step ahead ($h = 1$, i.e., setting $k = 1$ in Alg.~\ref{alg:agpm}). Their probing performance outperforms metric-based methods in GnuTella network $20\%$ on average. They match performance of other methods in Road network and are $7\%$ worse than \textsf{BC} in Collaboration network. The experiment with \textsf{MaxOutProbe} in Road network faces the problem of non-adaptive behavior of \textsf{MaxOutProbe}: $|V^p|$ of Road network's samples is smaller than maximal budget $k$, besides, due to the very low density property of Road network,  after probing $|V^p|$ nodes in sampled networks \textsf{MaxOutProbe} explores only few new nodes in underlying network; these factors lead to very poor performance of \textsf{MaxOutProbe}. \textsf{MaxOutProbe} also performs poorly as compared with adaptive implementation of metric based methods: It is $67\%$, and $36\%$ worse than \LINREG{} in Collaboration, GnuTella network respectively.

\subsection{Benefits of Looking Into Future Gain.}
For \LINREG{} and \LOGREG{}, we use different functions trained by different labels (marked as $h = 1$, $h = 2$, $h = 3$, $h = 4$ , which indicates regression functions are trained with benefits of node in 1-step, 2-step, 3-step or 4-step ahead). In Fig. \ref{fig:set3}, we evaluate \LINREG{} and \LOGREG{} with \textsf{RAND} and the best metric-based methods for each type of network reported in Table ~\ref{tbl:bnc}.  We use of \text{DEG} as baseline for our performance comparison. Specifically, we take the ratio of number of newly explored nodes of each probing method to number of newly explored nodes of \textsf{DEG} at the end of probing process. We omit result of \textsf{MaxOutProbe} due to its poor performance observed from previous subsection.

\LINREG{} shows consecutive improvement in probing performance with $h = 1$, $h = 2$, $h = 3$ as it ranks candidate nodes in set $V^p$ based on their predicted gain at increasing number of hops far away from them. Overall, \LINREG{} has better performance than \LOGREG{}. The $h = 3$ of \LINREG{} and \LOGREG{} outperforms \textsf{DEG} from $12\%$ to $15\%$ in Collaboration and GnuTella network. The result observed for Road network is from $1\%$ to $5\%$. \LINREG{} with $h = 3$ outperforms the best metric-based method in GnuTella network $11.6\%$; the improvement for Collaboration and Road network are $7.5\%$, $2.2\%$ respectively. The $h = 4$ of \LINREG{} and \LOGREG{} starts decreasing compared with $h = 1$, $h = 2$ and $h = 3$. This indicates the benefit of looking further benefits of selecting a node is true within specific number of hops far away from nodes. 

Among metric-based methods, while \textsf{BC} consistently performs well, the best method varies across the networks. It indicates the underlying network structure impacts performance of metric-based method and it is hard to determine which node centrality method is the best for which type of network. Interestingly, performance of random probing matches or even outperforms \textsf{BC}, \textsf{DEG}, \textsf{CC} in GnuTella networks. It is because GnuTella networks have low average clustering coefficient that makes BFS-based samples of these networks have star structure. Consequently, metric-based methods tend to rank candidate nodes in sampled network with same score. This helps random probing performs better than metric-based methods in this type of network.

\section{Conclusion}
\label{sec:con}
This paper studies the Graph Probing Maximization problem which serves as a fundamental component in many decision making problems. We first prove that the problem is not only NP-hard but also cannot be approximated within any finite factor. We then propose a novel machine learning framework to adaptively learn the best probing strategy in any individual network. The superior performance of our method over metric-based algorithms is shown by a set of comprehensive experiments on many real-world networks.
\bibliographystyle{ieeetr}
\bibliography{adasubmod,cs,webcrawler,socialcrawler,socialbot,linkpredict,approx,dataset,viral}

\begin{thebibliography}{10}

\bibitem{Ng16}
H.~T. Nguyen and T.~N. Dinh, ``Targeted cyber-attacks: Unveiling target
  reconnaissance strategy via social networks,'' in {\em INFOCOM WKSHPS},
  pp.~--, IEEE, 2016.

\bibitem{avr14y}
K.~Avrachenkov, P.~Basu, G.~Neglia, B.~Ribeiro, and D.~Towsley, ``Pay few,
  influence most: Online myopic network covering,'' in {\em INFOCOM WKSHPS},
  pp.~813--818, IEEE, 2014.

\bibitem{Soundarajan15}
S.~Soundarajan, T.~Eliassi-Rad, B.~Gallagher, and A.~Pinar, ``Maxoutprobe: An
  algorithm for increasing the size of partially observed networks,'' {\em
  NIPS}, 2015.

\bibitem{Hanneke09}
S.~Hanneke and E.~P. Xing, ``Network completion and survey sampling.,'' in {\em
  AISTATS}, pp.~209--215, 2009.

\bibitem{Masrour15}
F.~Masrour, I.~Barjesteh, R.~Forsati, A.~Esfahanian, and H.~Radha, ``Network
  completion with node similarity: A matrix completion approach with provable
  guarantees,'' in {\em ASONAM}, pp.~302--307, IEEE, 2015.

\bibitem{Golovin10}
D.~Golovin and A.~Krause, ``Adaptive submodularity: A new approach to active
  learning and stochastic optimization.,'' in {\em COLT}, pp.~333--345, 2010.

\bibitem{Seeman13}
L.~Seeman and Y.~Singer, ``Adaptive seeding in social networks,'' in {\em
  FOCS}, pp.~459--468, IEEE, 2013.

\bibitem{Cho98}
J.~Cho, H.~Garcia-Molina, and L.~Page, ``Efficient crawling through url
  ordering,'' 1998.

\bibitem{Chakrabarti02}
S.~Chakrabarti, {\em Mining the Web: Discovering knowledge from hypertext
  data}.
\newblock Elsevier, 2002.

\bibitem{Ester04}
M.~Ester, H.-P. Kriegel, and M.~Schubert, ``Accurate and efficient crawling for
  relevant websites,'' in {\em VLDB}, pp.~396--407, VLDB Endowment, 2004.

\bibitem{Chakrabarti99}
S.~Chakrabarti, M.~Van~den Berg, and B.~Dom, ``Focused crawling: a new approach
  to topic-specific web resource discovery,'' {\em Computer Networks}, vol.~31,
  no.~11, pp.~1623--1640, 1999.

\bibitem{Chau07}
D.~H. Chau, S.~Pandit, S.~Wang, and C.~Faloutsos, ``Parallel crawling for
  online social networks,'' in {\em WWW}, pp.~1283--1284, ACM, 2007.

\bibitem{Mislove07}
A.~Mislove, M.~Marcon, K.~P. Gummadi, P.~Druschel, and B.~Bhattacharjee,
  ``Measurement and analysis of online social networks,'' in {\em SIGCOMM},
  pp.~29--42, ACM, 2007.

\bibitem{Wwak10}
H.~Kwak, C.~Lee, H.~Park, and S.~Moon, ``What is twitter, a social network or a
  news media?,'' in {\em WWW}, pp.~591--600, ACM, 2010.

\bibitem{Boshmaf12}
Y.~Boshmaf, I.~Muslukhov, K.~Beznosov, and M.~Ripeanu, ``Key challenges in
  defending against malicious socialbots,'' in {\em USENIX}, 2012.

\bibitem{Elishar12}
A.~Elishar, M.~Fire, D.~Kagan, and Y.~Elovici, ``Organizational intrusion:
  Organization mining using socialbots,'' in {\em SocialInformatics},
  pp.~7--12, IEEE, 2012.

\bibitem{Elyashar13}
A.~Elyashar, M.~Fire, D.~Kagan, and Y.~Elovici, ``Homing socialbots: intrusion
  on a specific organization's employee using socialbots,'' in {\em ASONAM},
  pp.~1358--1365, ACM, 2013.

\bibitem{Fire14}
M.~Fire, R.~Goldschmidt, and Y.~Elovici, ``Online social networks: Threats and
  solutions,'' {\em Communications Surveys \& Tutorials, IEEE}, vol.~16, no.~4,
  pp.~2019--2036, 2014.

\bibitem{Paradise15}
A.~Paradise, A.~Shabtai, and R.~Puzis, ``Hunting organization-targeted
  socialbots,'' in {\em ASONAM}, pp.~537--540, ACM, 2015.

\bibitem{Kim11}
M.~Kim and J.~Leskovec, ``The network completion problem: Inferring missing
  nodes and edges in networks.,'' in {\em SDM}, vol.~11, pp.~47--58, SIAM,
  2011.

\bibitem{Maiya10}
A.~S. Maiya and T.~Y. Berger-Wolf, ``Sampling community structure,'' in {\em
  WWW}, pp.~701--710, ACM, 2010.

\bibitem{Nguyen10}
N.~P. Nguyen, Y.~Xuan, and M.~T. Thai, ``A novel method for worm containment on
  dynamic social networks,'' in {\em Military Communications Conference,
  2010-MILCOM 2010}, pp.~2180--2185, IEEE, 2010.

\bibitem{Dinh12}
T.~N. Dinh, D.~T. Nguyen, and M.~T. Thai, ``Cheap, easy, and massively
  effective viral marketing in social networks: truth or fiction?,'' in {\em
  Proceedings of the 23rd ACM conference on Hypertext and social media},
  pp.~165--174, ACM, 2012.

\bibitem{Nguyen11over}
N.~P. Nguyen, T.~N. Dinh, D.~T. Nguyen, and M.~T. Thai, ``Overlapping community
  structures and their detection on social networks,'' in {\em Privacy,
  Security, Risk and Trust (PASSAT) and 2011 IEEE Third Inernational Conference
  on Social Computing (SocialCom), 2011 IEEE Third International Conference
  on}, pp.~35--40, IEEE, 2011.

\bibitem{Dinh14}
T.~N. Dinh, H.~Zhang, D.~T. Nguyen, and M.~T. Thai, ``Cost-effective viral
  marketing for time-critical campaigns in large-scale social networks,'' {\em
  IEEE/ACM Transactions on Networking}, vol.~22, no.~6, pp.~2001--2011, 2014.

\bibitem{Newman10}
M.~Newman, {\em Networks: an introduction}.
\newblock OUP Oxford, 2010.

\bibitem{Vandin11}
F.~Vandin, E.~Upfal, and B.~J. Raphael, ``Algorithms for detecting
  significantly mutated pathways in cancer,'' {\em J. Comp. Biol.}, vol.~18,
  no.~3, pp.~507--522, 2011.

\bibitem{li05}
F.~Li, D.~Cheng, M.~Hadjieleftheriou, G.~Kollios, and S.-H. Teng, ``On trip
  planning queries in spatial databases,'' in {\em SSTD}, pp.~273--290,
  Springer, 2005.

\bibitem{brinkhoff02}
T.~Brinkhoff, ``A framework for generating network-based moving objects,'' {\em
  GeoInformatica}, vol.~6, no.~2, pp.~153--180, 2002.

\end{thebibliography}

\appendix
\section*{Supplementary Material for `\textit{Towards Optimal Strategy for Adaptive Probing in Incomplete Networks}'}
\subsection*{Complete proof of Theorem~\ref{theo:hardness}}
\begin{proof}
\balance
To prove Theorem~\ref{theo:hardness}, we construct classes of instances of the problems that there is no approximation algorithm with finite factor.

\begin{figure}[!ht]
	\vspace{-0.25in}
	\centering
    \subfloat[Without global degrees]{
 		\includegraphics[width=0.4\linewidth]{figures/hard_nd.pdf}
        \label{fig:hard_nodegree}
 	}
 	\subfloat[With global degrees]{
 		\includegraphics[width=0.4\linewidth]{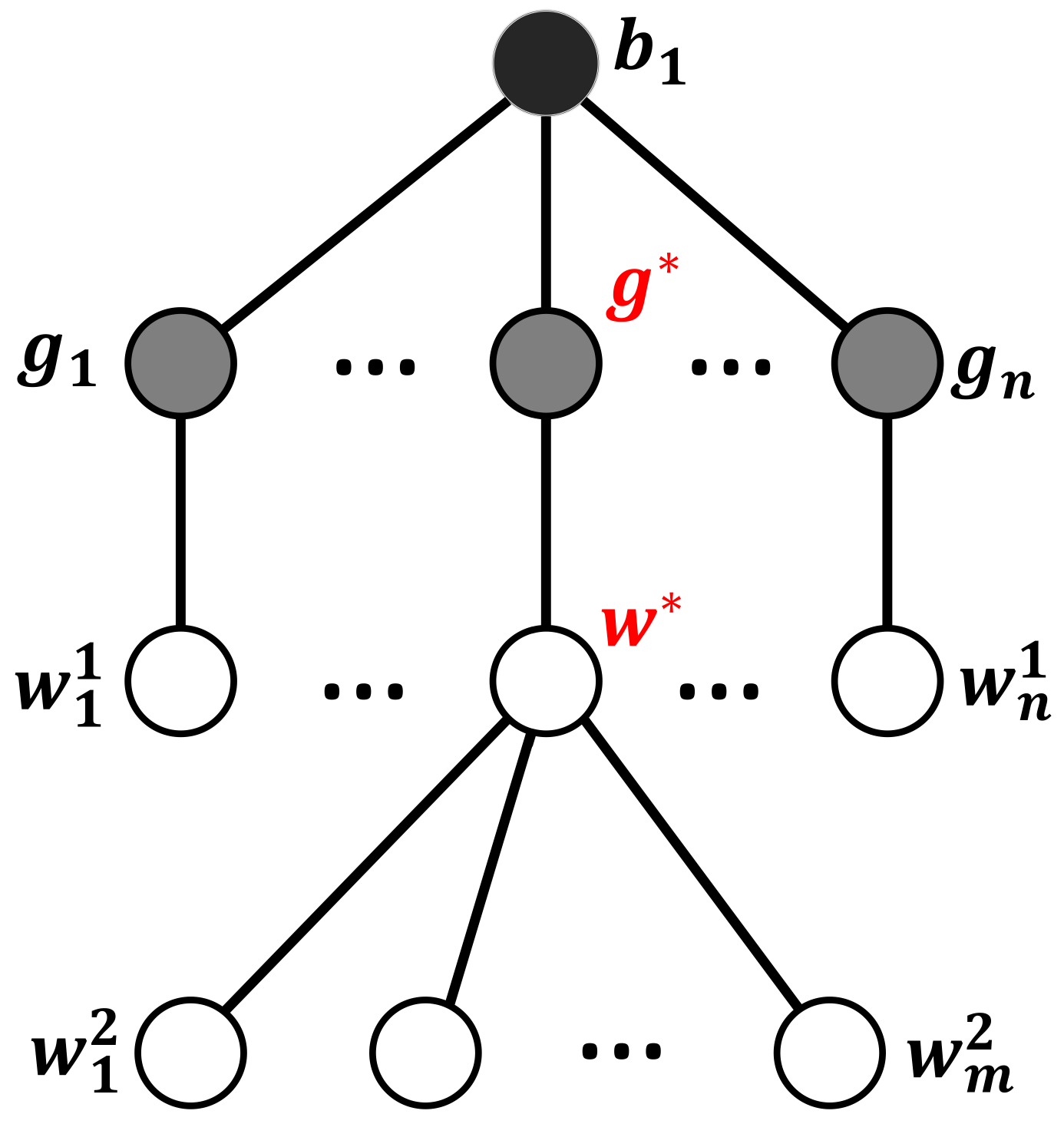}
		\label{fig:hard_degree}
 	}
	\vspace{-0.05in}
	\caption{Hardness illustration on global node degrees.}
	\vspace{-0.1in}
\end{figure}

We construct a class of instances of the probing problems as illustrated in Figure~\ref{fig:hard_nodegree}. Each instance in this class has: a single fully probed node in black $b_1$, $n$ observed nodes in gray each of which has an edge from $b_1$ and one of the observed nodes, namely $g^*$ varying between different instances, having $m$ connections to $m$ unknown nodes in white. Thus, the partially observed graph contains $n+1$ nodes, one fully probed and $n$ observed nodes which are selectable, while the underlying graph has in total of $n+m+1$ nodes. Each instance of the family has a different $g^*$ that $m$ unknown nodes are connected to. We now prove that in this class, no algorithm can give a finite approximate solution for the two problem.

First, we observe that for any $k \geq 1$, the optimal solution which probes the nodes with connections to unknown nodes has the optimal value of $m$ newly explore nodes, denoted by $OPT = m$. We sequentially examine two possible cases of algorithms, i.e., deterministic and randomized.
\begin{itemize}
	\item Consider a deterministic algorithm $\mathcal{A}$, since the $\mathcal{A}$ is unaware of the connections from gray to unknown nodes, given a budget $1 \leq k \ll n$, the lists or sequences of nodes that $\mathcal{A}$ selects are exactly the same for different instances of problems in the class. Thus, there are instances that $g^*$ is not in the fixed list/sequence of nodes selected by $\mathcal{A}$. In such cases, the number of unknown nodes explored by $\mathcal{A}$ is 0. Compared to the $OPT = m$, $\mathcal{A}$ is not a finite factor approximation algorithm.
	\item Consider a randomized algorithm $\mathcal{B}$, similarly to the deterministic algorithm $\mathcal{A}$, $\mathcal{B}$ does not know the connections from the partially observed nodes to white ones. Thus, the randomized algorithm $\mathcal{B}$ essentially selects at random $k$ nodes out of $n$ observed nodes. However, this randomized scheme does not guarantee to select $g^*$ as one of its selected nodes and thus, in many situations, the number of unknown nodes discovered is 0 that invalidates $\mathcal{B}$ to be a finite factor approximation algorithm. In average, $\mathcal{B}$ has $\frac{k}{n}$ chance of selecting $g^*$ which leads to an optimal solutions with $OPT=m$. Hence, the objective value is $\frac{km}{n}$ and the ratio with optimal value is $\frac{k}{n}$. Since $k \ll n$, we can say that the ratio is $O(\frac{1}{n})$ which is not finite in the average case for randomized algorithm $\mathcal{B}$.
\end{itemize}
In both cases of deterministic and randomized algorithms, there is no finite factor approximation algorithm for \AGPM{} or \BGPM{}.
\end{proof}

\balance
\end{document}